\documentclass[journal]{IEEEtran}

\usepackage{booktabs}
\usepackage{dsfont}
\usepackage{stmaryrd}
\usepackage{amsmath}
\usepackage{amssymb}
\usepackage{amsthm}
\usepackage{url}
\usepackage[compress]{cite}
\usepackage{tikz}
\usepackage{psfrag}

\newtheorem{assumption}{Assumption}[section]
\newtheorem{theorem}{Theorem}[section]
\newtheorem{definition}{Definition}[section]
\newtheorem{corollary}{Corollary}[section]

\newtheorem{proposition}{Proposition}[section]
\newtheorem{remark}{Remark}[section]
\newtheorem{example}{Example}[section]

\newcommand{\range}[1]{\llbracket #1 \rrbracket}
\newcommand{\revise}[1]{{\color{black}#1}}

 \usetikzlibrary{patterns}

\DeclareMathOperator*{\essup}{ess\,sup}
\DeclareMathOperator*{\essinf}{ess\,inf}
\DeclareMathOperator*{\argmax}{arg\,max}
\DeclareMathOperator*{\argmin}{arg\,min}
\DeclareMathOperator*{\diam}{diam}

\begin{document}

\title{Development and Analysis of Deterministic Privacy-Preserving Policies Using  Non-Stochastic Information Theory
\thanks{The work of F. Farokhi was partially supported by the McKenzie Fellowship at the University of Melbourne and the VESKI Victoria Fellowship by the Victorian State Government. He also would like to thank the office of the Deputy Vice-Chancellor (Research) at the University of Melbourne for funding his current fellowship position at the university.  }
}

\author{Farhad Farokhi
\thanks{F. Farokhi is with the CSIRO's Data61 and Department of Electrical and Electronic Engineering at the University of Melbourne. e-mail: farhad.farokhi@\{unimelb.edu.au,data61.csiro.au\}}
\vspace{-.2in}
}

\maketitle

\begin{abstract} A \revise{deterministic} privacy metric using non-stochastic information theory is developed. Particularly, maximin information is used to construct a measure of information leakage, which is inversely proportional to the measure of privacy. Anyone can submit a query to a trusted agent with access to a non-stochastic uncertain private dataset. Optimal deterministic privacy-preserving policies for responding to the submitted query are computed by maximizing the measure of privacy subject to a constraint on the worst-case quality of the response (i.e., the worst-case difference between the response by the agent and the output of the query computed on the private dataset). The optimal privacy-preserving policy is proved to be a piecewise constant function in the form of a quantization operator applied on the output of the submitted query. The measure of privacy is also used to analyze $k$-anonymity (a popular deterministic mechanism for privacy-preserving release of datasets using suppression and generalization techniques), proving that it is in fact not privacy-preserving. 
\end{abstract}


\section{Introduction}
Advances in communication and computation engineering have enabled the use of big data analysis for answering societal challenges. These advances have motivated incorporation of new tools for collection and analysis of datasets, and reporting data-driven insights. The erosion of privacy caused by the adoption of such tools has resulted in adoption of new rules by governments, such as the General Data Protection Regulation (GDPR) in the European Union, for protecting citizens, customers, and their data. 

Anonymization is most often used as a method of choice by governments or companies alike for releasing private datasets\footnote{See \url{https://data.gov.au} for an example of government initiave. Many other examples can be found in \url{https://www.kaggle.com}.} to the broader public for analysis. Although popularly adopted, anonymization has been proved to be insufficient for privacy preservation~\cite{narayanan2008robust,su2017anonymizing, de2013unique}. Therefore, systematic methods  for  privacy preservation in a provable manner should be developed.

Differential privacy and its variants, such \revise{as} local differential privacy and probabilistic differential privacy, form a category of methodologies with provable privacy guarantees~\cite{dwork2008differential,dwork2014algorithmic, duchi2013local,kairouz2014extremal,machanavajjhala2008privacy, hall2012random,padakandla2018preserving}. These methods, in summary, rely on the use of randomized policies, such as additive noise, to ensure that the statistics of the reported outputs do not change noticeably by variations in an individual entry of the dataset. This property ensures that an adversary cannot reverse-engineer differentially-private outputs to accurately estimate an individual private entry of the dataset, even in the presense of side information. Various studies have been devoted to finding ``optimal'' noise distribution in differential privacy~\cite{soria2013optimal,geng2014optimal, geng2016optimal}; however, off-the-shelf mechanisms, such as the additive Laplace and Gaussian noise with scales proportional to the sensitivity of the submitted query with respect to \revise{the} individual entries of the dataset, are often used to ensure differential privacy~\cite{dwork2014algorithmic}. Note that the use of randomized policies for privacy protection in itself is not particularly new~\cite{warner1965randomized} but, prior to differential privacy, provable guarantees were often missing. 

Another methodology for privacy protection is the use information theoretic metrics dating back to the pioneering work on secrecy in~\cite{6772207}. In \revise{the} secrecy problem, a sender wishes to devise an encoding scheme to create a secure channel for communicating with a receiver while hiding her data from an eavesdropper (similar to the setup of encryption). The privacy problem with the emphasis on masking or equivocating of information from the intended primary receiver (rather than an eavesdropper) or a secondary receiver with as much information as the primary reciever have been studied in~\cite{sankar2013utility,courtade2012information,yamamoto1983source, yamamoto1988rate}.  Information-theoretic guarantees have been also provided on the amount of leaked private information when utilizing differential privacy~\cite{Aceto2011,du2012privacy}. Furthermore, entropy, mutual information, Kulback-Leiber divergence, and Fisher information have been  repeatedly used as measure of privacy in~\cite{farokhi2015quadratic,farokhi2016privacy,wainwright2012privacy, liang2009information,lai2011privacy,li2015privacy,bassi2018lossy, farokhi2018fisher}. 

A common thread or assumption among all these methodologies is that they utilize randomization for safeguarding privacy. In fact, the definition of differential privacy assumes the use of randomized functions and information theoretic tools used so far have been based on randomized random variables. However, many popular\footnote{\revise{Popularity of these methods is somehwat evident from the sheer number of available toolboxes for implementation~\url{https://arx.deidentifier.org/overview/related-software/}}} heuristic-based privacy-preserving methods, such as $k$-anonymity~\cite{samarati2001protecting,sweeney2002k} and $\ell$-diversity~\cite{1617392}, are deterministic \revise{(i.e., deterministic mappings, such as suppression and generalization, applied to non-stochastic datasets)}.  

\revise{Randomized, or stochastic, privacy-preserving policies have been shown to cause problems, such as un-truthfulness~\cite{bild2018safepub}, which can be undesirable in practice~\cite{Poulis2015}. This is perhaps one of the reason behind low popularity of randomized privacy-preserving policies, such as differential privacy, within the financial or health sectors~\cite{bild2018safepub}. For instance, randomized privacy-preserving policies in financial auditing have been criticized for complicating fraud detection~\cite{bhaskar2011noiseless,nabar2006towards}. Also, generation of unreasonable and unrealistic outputs by randomness can cause undesirable financial outcomes (e.g., missleading investors or market operators by reporting noisy outputs that point to lack of liquidity in a bank). Randomized privacy-preserving policies, in general, have also encountered difficulties in medical, health, or social sciences~\cite{dankar2013practicing,Mervis114}. Finally, undesirable properties of differentially-private additive noise, especially the Laplace noise, might make it less appealing. For instace, optimal variable estimation in the presence of privacy-preserving Laplace noise is computationally expensive~\cite{farokhi2016optimal} and probability of returning impossible reports (e.g., negative median income) could be relatively high due to slow-decaying nature of Laplace noise~\cite{bambauer2013fool}.} 

\revise{In addition to the aforementioned difficulties or negative consequences associated with randomized policies, the popularity of non-stochastic methods might also be caused by the simplicity of implementing deterministic policies, in the sense of not requiring a working knowledge of random variables and their generation by laymen.  Deterministic privacy-preserving policies and non-stochastic measures of information leakage, if designed correctly, can also provide concrete guarantees regarding the amount of the information that can be inferred about each instance of the private dataset, rather than stochastic measures of privacy that only provide guarantees in average (i.e., in a statistical sense).}

So far, deterministic privacy-preserving policies are generated in an \textit{ad hoc} manner and are often vulnerable to attacks (e.g., $k$-anonimity has been proved to be \revise{vulnerable} to attacks, such as homogeneity attack~\cite{1617392}). This is because there is no good measure of privacy that works for deterministic policies on deterministic datasets. Therefore, one cannot prove (in some sense) privacy guarantees of the methods (even if weak or limited in scope or practice). The  popularity of non-stochastic privacy-preserving policies justifies requiring a metric for their analysis and comparison (irrespective of their inherent philosophical weaknesses in comparison to stochastic policies).

Motivated by this observation, in this paper, a deterministic privacy metric based on non-stochastic information theory is developed. Traditional information theory, starting with Shannon's seminal work in~\cite{shannon1948mathematical}, usually assumes that data (source) and communication channels are stochastic in nature. This has been proved to be extremely powerful in modelling and analysing communication systems; see, e.g.,~\cite{Elements2006} and references there-in. However, the notion of information within the traditional information theory literature, such as mutual information, is not useful for analysing non-stochastic uncertain variables \revise{(an analogue of random variables but without a probability measure)} and deterministic privacy-preserving policies. \revise{This is because such definitions require a probability density function to exist for variables, which is not the case in the absence of additive privacy-preserving noise (with a known probability density function) or stochasticity assumptions on the private dataset.  }

\revise{T}here is a parallel less-studied (within tertiary colleges) theory of non-stochastic information theory~\cite{hartley1928transmission, kolmogorov1959varepsilon,renyi1961measures, nair2013nonstochastic,jagerman1969varepsilon}, which has been recently used within engineering~\cite{nair2012nonstochastic, duan2015transfer,wiese2016uncertain}. \revise{Non-stochastic information theory relies on uncertain variables and extension of analogues of probabilistic ideas, such as independence. Non-stochastic information theory is not equivalent to treating input variables with known, bounded ranges as uniformly distributed random variables because such an approach is still probabilistic, and the output random variables may exhibit non-uniform distributions despite the uniform inputs. In contrast, in the uncertain variable model, only the support sets are considered, and no distributions are derived at any stage. }

In this paper,  non-stochastic measures of information, such as maximin information, from the non-stochastic \revise{information} theory literature are used to develop a measure of privacy. Anyone can submit a query to a trusted agent with access to a non-stochastic uncertain private dataset. \revise{An optimization problem is posed to maximize the measure of privacy subject to a constraint on the worst-case quality of the response (i.e., the worst-case difference between the response by the agent and the output of the query computed on the private dataset). The solution to the optimization problem captures the optimal deterministic privacy-preserving policies for responding to submitted queries.}  The optimal privacy-preserving policy is in fact \revise{proved to be} a quantization operator applied on the output of the submitted query computed based on the private dataset. The developed measure of privacy is utilized to analyze the \revise{privacy credentials} of $k$-anonymity, proving that it is not privacy-preserving, which was previously observed using \revise{adversarial} attacks in~\cite{1617392}. 

The rest of the paper is organized as follows. Section~\ref{sec:nonstoinfo} provides a summary of non-stochastic information theory. The problem formulation is presented in \revise{Section~\ref{sec:problem}}. In Section~\ref{sec:privacypolicy}, a piecewise constant function, in the form of a quantization operator, is proved to be an optimal privacy-preserving policy. The privacy of $k$-anonymity is analyzed using the proposed non-stochastic privacy metrics in Section~\ref{sec:kanonimity}. Finally, Section~\ref{sec:conclusions} concludes the paper and presents future directions for work.

\section{Non-stochastic Information Theory} \label{sec:nonstoinfo}
\revise{In this section, an overview of non-stochastic information theory is presented. First, uncertain variables, which are analogues of random variables but without a probability measure, are introduced. Then, various measures of information, i.e., non-stochastic information based on R\'{e}nyi differential $0$-entropy, non-stochastic information leakage, and maximin information are presented. }

\subsection{Uncertain Variables}
Consider sample space $\Omega$. Each element $\omega\in \Omega$ is referred to as a sample. The sample space is the source of uncertainty. Any mapping $X:\Omega\rightarrow\mathbb{X}$ defines an uncertain variable. A realization of such a variable is $X(\omega)$. However, for the ease of presentation, $X(\omega)$ is replaced by $X$ when the dependence of the uncertain variable to the sample is evident from the context. Up to this point, the difference between uncertain variables and random variables is the absence of a measure on the space $\Omega$. Throughout this paper, it is assumed that all uncertain variables are real-valued, i.e., $\mathbb{X}\subseteq\mathbb{R}^{n_x}$ for some $n_x\in\mathbb{N}$. Marginal range of $X$ is defined as 
$$
\range{X}:=\{X(\omega):\omega\in\Omega \}\subseteq \mathbb{X}.
$$
Joint range of two uncertain variables $X:\Omega\rightarrow\mathbb{X}$ and $Y:\Omega\rightarrow\mathbb{Y}$ is 
$$
\range{X,Y}:=\{(X(\omega),Y(\omega)):\omega\in\Omega \}\subseteq \mathbb{X}\times \mathbb{Y}.
$$ 
Finally, conditional range of $X$ (conditioned on the observation of another uncertain variable $Y(\omega)=y$) is given by  
$$
\range{X|y}:=\{X(\omega):\exists \omega\in\Omega  \mbox{ \revise{such that }} Y(\omega)=y\}\subseteq \range{X}. 
$$
The family of all conditional ranges is denoted by 
$$
\range{X|Y}:=\{\range{X|y}:y\in \range{Y}\}\subseteq 2^{\range{X}}.
$$
This should not be mistaken with the union of all such conditional ranges given by $\bigcup_{y\in \range{Y}}\range{X|y}= \range{X}$. In fact, regarding the union, it can be proved that $\bigcup_{y\in \range{Y}} \range{X|y}\times\{y\}=\range{X,Y}.$ 

\begin{definition}[Unrelatedness]
Uncertain variables $X_i$, $i=1,\dots,n$, are unrelated if $\range{X_1,\dots,X_n}=\range{X_1}\times \cdots \times \range{X_n}.$ Further, they are  conditionally unrelated (conditional on $Y$) if 
$\range{X_1,\dots,X_n|y}=\range{X_1|y}\times \cdots \times \range{X_n|y}$ for all $y\in\range{Y}$.
\end{definition}

For two uncertain variables, this definition is equivalent to stating that $X_1$ and $X_2$ are unrelated if $\range{X_1|x_2}=\range{X_1}, \forall x_2\in\range{X_2},$ and \textit{vice versa}. 
Again, for two uncertain variables, this definition is equivalent as saying that $X_1$ and $X_2$ are conditionally unrelated (conditional on $Y$) if
$\range{X_1|x_2,y}=\range{X_1|y}, \;\forall (x_2,y)\in\range{X_2,Y},$
Finally, for uncertain variables $X$ and $Y_i$, $i=1,\dots,n$, it can be seen that $\range{X|y_1,\dots,y_n}\subseteq \bigcap_{i=1}^n \range{X|y_i},$
where the inequality is achieved\revise{, i.e., $\range{X|y_1,\dots,y_n}= \bigcap_{i=1}^n \range{X|y_i}$,} if $Y_i$, $i=1,\dots,n$, are unrelated conditional on $X$.

\subsection{Non-stochastic Entropy and Information}
The non-stochastic entropy of uncertain variable $X$ can be defined as
\begin{align} \label{eqn:firstentropys}
h_0(X):=\log(\mu(\range{X}))\in \overline{\mathbb{R}},
\end{align}
where \revise{$\mu$ is the Lebesgue measure,} $\overline{\mathbb{R}}$ is the extended real line $\mathbb{R}\cup\{\pm \infty\}$\revise{,} and the logarithm can be taken in any basis. In line with the differential entropy for random variables, the logarithm is in the natural basis throughout the rest of the paper. The non-stochastic entropy in~\eqref{eqn:firstentropys} is sometimes referred to as R\'{e}nyi differential $0$-entropy~\cite{nair2013nonstochastic}.

\begin{remark}[$\varepsilon$-entropy]
This notion of R\'{e}nyi differential $0$-entropy is intimately related to the $\varepsilon$-entropy~\cite{kolmogorov1959varepsilon} defined as $h_\epsilon(X):=\log(N_\varepsilon(\range{X})),$ where $N_\epsilon(\cdot)$ is the smallest number of sets of diameter $2\varepsilon$ \revise{such} that their union covers $\range{X}$, referred to as the minimal $\varepsilon$-covering. The inequality $\varepsilon^{n_x} N_\varepsilon(\range{X})\leq \mu(\range{X})\leq (2\varepsilon)^{n_x} N_\varepsilon(\range{X})$ implies that $0 \leq h(X)-[h_\varepsilon(X)+{n_x} \log(\varepsilon)]\leq {n_x} \log(2).$ These two notions of entropy are similar. 
\end{remark}

Similarly,  the non-stochastic relative (or conditional) entropy of uncertain variable $X$ conditioned on  uncertain variable  $Y$ can be defined as
\begin{align}
h_0(X|Y):=\essup_{y\in \range{Y} } \log(\mu(\range{X|y})),
\end{align}
where, for any real-valued function $f:\mathcal{X}\rightarrow\mathbb{R}$ for some $\mathcal{X}\subseteq\mathbb{R}^m$, the essential supremum is defined as
\begin{align*}
\essup_{x\in \mathcal{X}} f(x):=\inf\{b\in\mathbb{R}:\mu(\{x\in\mathcal{X}:f(x)> b\})=0\}.
\end{align*}
Based on the definition of entropy, the non-stochastic information between two uncertain variables $X$ and $Y$ can also be defined as
\begin{align}
I_0(X;Y):=&h_0(X)-h_0(X|Y)\nonumber\\
=&\essinf_{y\in \range{Y} } \log\left(\frac{\mu(\range{X})}{\mu(\range{X|y})} \right).\label{eqn:information1}
\end{align} 
Note that Kolmogorov had defined `combinatorial' conditional entropy using $\log(\mu(\range{X|y}))$ and the measure of information gain was defined as $\mu(\range{X})/\mu(\range{X|y})$ in~\cite{kolmogorov1959varepsilon}. These quantities are only defined for an observed value of uncertain variable $Y = y$; however, the definition in~\eqref{eqn:information1} relies on the worst-case ratio. 

Now, a non-stochastic version of Fano's inequality can be established. Let the uncertain variable $\hat{X}(y)$ denote an estimate of an uncertain variable $X$ based on uncertain variable $Y$ for measurement $Y=y$. In this paper,  only  unbiased estimators, defined below, are considered.

\begin{assumption}[Unbiased Estimator] \label{assum:unbiased} An estimator $\hat{X}:\range{Y}\rightarrow\range{X}$ is unbiased if $\hat{X}(y)\in \range{X|y}$. 
\end{assumption}

This essentially means that the estimate is consistent with the received measurement, i.e., $X,\hat{X}(y)\in \range{X|y}$. A measure of the quality of the estimate can be defined as
\begin{align}
d_{\max}(X,\hat{X}(Y)):=&\essup_{y\in\range{Y}} \essup_{x\in\range{X|y}}\|x-\hat{X}(y)\|_2.
\end{align}
This measure captures the largest worst-case distance between uncertain variable $X$ and its estimate. Before stating the following theorem, a notation needs to be defined. Let $\Gamma:z\mapsto \int_0^\infty x^{z-1}\exp(-x)\mathrm{d}x$ be the Gamma function (extension of factorial to real numbers).

\begin{theorem} \label{tho:Fano} Consider $X$ and $Y=f(X)$ are uncertain variables for some function $f:\range{X}\rightarrow\range{Y}$. Assume that $\range{X|y}$ is a Borel set for all $y\in\range{Y}$. Then,
\begin{align*}
\frac{\Gamma(n_x/2+1)^{1/n_x}}{\sqrt{\pi}}\exp\bigg( \frac{h_0(X|Y)}{n_x} \bigg)&\leq d_{\max}(X,\hat{X}(Y)).
\end{align*}
\end{theorem}

\begin{proof} 
Further note that
\begin{align*}
\essup_{y\in\range{Y}} \essup_{x\in\range{X|y}}&\|x-\hat{X}(y)\|_2\\
\geq &\essup_{y\in\range{Y}}\inf_{\hat{X}}\essup_{x\in\range{X|y}} \|x-\hat{X}(y)\|_2
\\
\geq &\essup_{y'\in\range{Y}} \frac{1}{2}\diam(\range{X|y'}),
\end{align*}
where the last inequality follows from the fact that $\essup_{x\in\range{X|y'}} \|x-\hat{X}(y')\|_2$ is the radius of a ball that encompasses $\range{X|y'}$ and is centred at $\hat{X}(y')\in\range{X|y'}$ (see Assumption~\ref{assum:unbiased}) and the smallest such radius is always larger than or equal to half of the diameter. Therefore,
\begin{align*}
\essup_{x\in\range{X}} \essup_{y:x\in\range{X|y}}\|x-\hat{X}(y)\|_2
\geq \frac{1}{2} \essup_{y\in\range{Y}}\mathfrak{H}^{n_x}(\range{X|y}).
\end{align*}
The last step follows from the relationship between the Hausdorff measure $\mathfrak{H}^{n_x}(\cdot)$ and the Lebesgue measure $\mu(\cdot)$ for Borel sets~\cite[p\,28-30]{ambrosio2004topics}. This completes the proof.
\end{proof}

\begin{example}\label{example:1}
The notions of non-stochastic information and relative entropy are not useful for measuring privacy leakage. This is because it considers the \revise{size of the largest} $\mu(\range{X|y})$, while privacy wants to ensure that all $\mu(\range{X|y})$ are large. To see this, consider the following example:
\begin{align*}
f(X):=
\begin{cases}
X, & 0\leq X< 1/2,\\
1, & \mbox{otherwise},
\end{cases}
\end{align*}
where $X$ is an uncertain variable with $\range{X}=[0,1]$. It is easy to show that $h_0(X|f(X))=\log(1/2)$\revise{; note that $h_0(X)=\revise{0}$}. Construct an estimator of the form
\begin{align*}
\hat{X}(Y):=
\begin{cases}
Y, & 0\leq Y< 1/2,\\
3/4, & \mbox{otherwise},
\end{cases}
\end{align*}
Note that $d_{\max}(X,\hat{X}(f(X)))=1/4$ attaining the lower bound in Theorem~\ref{tho:Fano} (as $\Gamma(3/2)=\sqrt{\pi}/2$), proving that $\hat{X}(\cdot)$ is optimal in the sense of minimizing $d_{\max}(X,\hat{X}(f(X)))$. The function $f(\cdot)$ is clearly not privacy-preserving as $f(X)=X$ for many inputs! In fact, $\inf_{y\in\range{Y}}\mu(\range{X|y})=0$. 
\end{example}

Therefore, a notion of \textit{relative disarray} can be defined:
\begin{align}
d_0(X|Y):=\inf_{y\in\range{Y}}\log(\mu(\range{X|y})).
\end{align}
Following this, \textit{non-stochastic information leakage} can be defined as
\begin{align}
L_0(X;Y):=h_0(X)-d_0(X|Y).
\end{align}
Another useful measure of the quality of an estimator is
\begin{align}
d_{\min}(X,\hat{X}(Y)):=\essinf_{y\in\range{Y}} \essup_{x\in\range{X|y}}\|x-\hat{X}(y)\|_2.
\end{align}
This measure captures the smallest worst-case distance between uncertain variable $X$ and its estimate. If $d_{\min}(X,\hat{X}(Y))$ is small, it means that there exist some values for uncertain variable $X$ for which the privacy is not preserved in the sense that an adversary can reconstruct $X$ for those values accurately based on $Y$. 

\begin{theorem} \label{tho:Fano_1} Consider $X$ and $Y=f(X)$ are uncertain variables for some function $f:\range{X}\rightarrow\range{Y}$. Assume that $\range{X|y}$ is a Borel set for all $y\in\range{Y}$. Then,
\begin{align*}
\frac{\Gamma(n_x/2+1)^{1/n_x}}{\sqrt{\pi}}\exp\bigg( \frac{d_0(X|Y)}{n_x} \bigg)
&\leq d_{\min}(X,\hat{X}(Y)).
\end{align*}
\end{theorem}

\begin{proof} The proof follows the same line of reasoning as in the proof of Theorem~\ref{tho:Fano}. Note that, 
\begin{align*}
\essinf_{y\in\range{Y}} \essup_{x\in\range{X|y}}&\|x-\hat{X}(y)\|_2\\
\geq &\essinf_{y\in\range{Y}}\essinf_{\hat{X}} \essup_{x\in\range{X|y}}\|x-\hat{X}(y)\|_2\\
\geq  &\essinf_{y\in\range{Y}} \frac{1}{2}\diam(\range{X|y}).
\end{align*}
This completes the proof.
\end{proof}

\setcounter{example}{0}
\begin{example}[Cont'd]
In this example, $d_0(X|f(X))=-\infty$ (by the convention that $\log(0)=\lim_{t\searrow 0}\log(t)=-\infty$) and $L_0(X;f(X))=+\infty$. Hence, non-stochastic information leakage $L_0(X;f(X))$ can accurately capture the fact that $f(X)$ is not privacy preserving. In addition, it can be seen that $d_{\min}(X,\hat{X}(Y))=0$, which proves that again $\hat{X}(Y)$ is optimal in the sense of the cost function $d_{\min}(X,\hat{X}(Y))$ (as the lower bound in Theorem~\ref{tho:Fano_1} is achieved). 
\end{example}

In general, the non-stochastic information $I_0(\cdot;\cdot)$ and non-stochastic information leakage $L_0(\cdot;\cdot)$ are not symmetrical, that is, $I_0(X;Y)\neq I_0(Y;X)$ and $L_0(X;Y)\neq L_0(Y;X)$ in general (contrary to mutual information in the information theory literature). \revise{Many measures of information have been introduced in the past for stochastic variables that are  asymmetric and have been proved to be useful in practice~\cite{massey1994guessing,smith2009foundations,braun2009quantitative, m2012measuring}. However, symmetry enables proving useful relaxations; see Proposition~\ref{prop:relaxed} in the next section.} A non-stochastic information transmission was proposed in~\cite{klir2005uncertainty}, defined as
\begin{align}
T_0(X;Y):=h_0(X)+h_0(Y)-h_0(X,Y).
\end{align}
This new measure of information is symmetric, that is, $T_0(X;Y)=T_0(Y;X)$.  Although being symmetric in general, utilization of this measure is not possible (because $\range{Y}$ can be a discrete set $\mu(\range{Y})=0$ and thus $h_0(X,Y)=0$ in all such cases). Another symmetric measure of information is the maximin information. In order to define this measure of information, the notion of taxicab connectivity must be defined\revise{, borrowed from the pioneering works in~\cite{nair2012nonstochastic,nair2013nonstochastic} on non-stochastic information theory.}

\begin{definition}[Taxicab Connectivity]\hfill\break\vspace{-1em}
\begin{itemize}
\item $(x,y),(x',y')\in\range{X,Y}$ are taxicab\footnote{The term refers to taxis/cabs in New York in which they connect two intersections by a sequence of horizontal or vertical moves.} connected if there exists a sequence of points $\{(x_i,y_i)\}_{i=1}^n\subseteq\range{X,Y}$ such that $(x_1,y_1)=(x,y)$, $(x_n,y_n)=(x',y')$, and either $x_i=x_{i-1}$ or $y_i=y_{i-1}$ for all $i\in\{2,\dots,n\}$;
\item $\mathcal{A}\subseteq\range{X,Y}$ is taxicab connected if all points in $\range{X,Y}$ are taxicab connected;
\item \revise{$\mathcal{A},\mathcal{B}\subseteq\range{X,Y}$ are taxicab isolated if there do not exist points $(x,y)\in\mathcal{A}$ and $(x',y')\in\mathcal{B}$ such that $(x,y)$ and $(x',y')$ are taxicab connected;}
\item A taxicab partition of $\range{X,Y}$ is a set of sets $\mathfrak{F}(X,Y):=\{\mathcal{A}_i\}_{i=1}^n$ such that $\range{X,Y}\subseteq\bigcup_{i=1}^n\revise{\mathcal{A}_i}$, any  $\mathcal{A}_i,\mathcal{A}_j$ are  taxicab isolated if $j\neq i$, and $\mathcal{A}_i$ is taxicab connected for all $i$.
\end{itemize}
\end{definition}

There exists a unique taxicab partition for any $\range{X,Y}$~\cite{nair2013nonstochastic}. Maximin information can be defined as
\begin{align}
I_\star(X;Y):=\log(|\mathfrak{F}(X,Y)|),
\end{align}
where $\mathfrak{F}(X,Y)$ denotes the unique taxicab partition of $\range{X|Y}$. It has been proved that $|\mathfrak{F}(X,Y)|=|\mathfrak{F}(Y,X)|$ and thus $I_\star(X;Y)=I_\star(Y;X)$ resulting in a symmetric notion of information~\cite{nair2013nonstochastic}. 

\setcounter{example}{0}
\begin{example}[Cont'd]
In this example, $I_\star(X;f(X))=+\infty$. This instantly shows that $f(X)$ is not privacy preserving.
\end{example}

\revise{\section{Problem Formulation} \label{sec:problem}}
 In what follows, it is assumed that a private dataset $X$ is available to a secure trusted agent. Anyone may submit a query of the form $f(\cdot)$, i.e., it can request that the trusted agent compute  and provide the response $f(X)$. 

\begin{definition}[Measure of Privacy] Let $\revise{\tilde{f}}(\cdot)$ be a reporting function and define uncertain variable $Y$ based on uncertain variable $X$ such that $Y=\revise{\tilde{f}}(X)$. Then, the measure of privacy for the reporting function $\revise{\tilde{f}}$ is
\begin{subequations}\label{eqn:define:P12}
\begin{align} \label{eqn:define:P1}
\revise{\mathcal{P}}_1(\revise{\tilde{f}}):=&\frac{1}{L_0(X;Y)},\\
 \label{eqn:define:P2}
\revise{\mathcal{P}}_2(\revise{\tilde{f}}):=&\frac{1}{I_\star(X;Y)}.
\end{align}
\end{subequations}
\end{definition}

The inverse relationship  between the measures of privacy and information in~\eqref{eqn:define:P12} is because information leakage reduces the privacy guarantee. A useful and intuitive property for the aforementioned measures of privacy can be proved to illustrate that after releasing an output it is not possible to gain more information from the data by additional manipulations.

\begin{theorem}[Post Processing] $\revise{\mathcal{P}}_i(g\circ f)\geq \revise{\mathcal{P}}_i(f)$ \revise{for both $i=1,2$}.
\end{theorem}

\begin{proof} Let uncertain variable $Y$ and $Z$ be defined as $Y(\omega):=f(X(\omega))$ and $Z(\omega):=g(Y(\omega))$ for all $\omega\in\Omega$.  The data processing inequality in~\cite{nair2013nonstochastic} implies that $I_\star(X;Z)\leq I_\star(X;Y)$. Therefore, $\revise{\mathcal{P}}_2$ can  be only increased by post processing. For the other measure of privacy note that
\begin{align*}
d_0(Z|X)=&\essinf_{z\in\range{Z}}\mu(\range{X|z})\\
=&\essinf_{z\in\range{Z}}\mu\bigg(\bigcup_{y'\in \range{Y|z}} \range{X|y'}\bigg)\\
\geq &\essinf_{z\in\revise{\range{Z}}}\essinf_{y\in\range{Y|z}}\mu(\range{X|y}) \\
=&\essinf_{y\in\range{Y}}\mu(\range{X|y}) \\
\geq &d_0(X|Y).
\end{align*}
Hence, $\revise{\mathcal{P}}_1$ can also  be only increased by post processing. This concludes the proof.
\end{proof}

The best policy for preserving privacy, maximizing the measure of privacy, is to ensure that $X$ and $f(X)$ are unrelated (making $\revise{\mathcal{P}}_i(f)=0$). This is, of course, without any value as all the information regarding $X$ would be lost and the utility of the report (in every possible sense) is zero. Therefore, there is a need for balancing utility and privacy. 

\begin{definition}[Measure of Quality] The measure of quality for the reporting function $\revise{\tilde{f}}$ for the query $f$ is 
\begin{align}\label{eqn:define:Q}
\mathfrak{Q}(\revise{\tilde{f}}):=\frac{1}{\displaystyle \essup_{x\in\range{X}}\linebreak\|f(x)-\revise{\tilde{f}}(x)\|_2}.
\end{align}
\end{definition}

With these definition ready, the optimal privacy-preserving policy $\tilde{f}$ can be computed by solving the optimization problem in
\begin{subequations} \label{eqn:Problem1}
\begin{align}
\mathbf{P}_\gamma: \max_{\revise{\tilde{f}}\in\mathcal{F}} \quad & \revise{\mathcal{P}}_i(\revise{\tilde{f}}),\\
\mathrm{s.t.}\quad & \mathfrak{Q}(\revise{\tilde{f}})\geq \gamma,
\end{align}
\end{subequations}
where $\mathcal{F}$ denotes the set of functions over which the privacy measure is optimized, i.e., the set of functions of interest for implementing as potential privacy-preserving policies. 

\begin{proposition}  \label{prop:relaxed} Assume that $\range{Y}\subseteq\mathbb{R}$. For any $\revise{\tilde{f}}=g\circ f$, 
\begin{subequations}
\begin{align}
I_\star(X;g(f(X)))&\leq I_\star(f(X);g(f(X))),\label{eqn:dataprocessing2}\\
\mathfrak{Q}(\revise{\tilde{f}})&=1/\essup_{y\in\range{Y}}\|y-g(y)\|_2.\label{eqn:dataprocessing3}
\end{align}
\end{subequations}
\end{proposition}

\begin{proof} 
Let uncertain variable $Y$ and $Z$ be defined as $Y(\omega):=f(X(\omega))$ and $Z(\omega):=g(Y(\omega))$ for all $\omega\in\Omega$. The data processing inequality~\cite{nair2013nonstochastic} shows that $I_\star(X;Z)\leq \min(I_\star(X;Y),I_\star(Z;Y))$. This concludes the proof for~\eqref{eqn:dataprocessing2}. The proof for~\eqref{eqn:dataprocessing3} follows from the definition. 
\end{proof} 

\revise{Proposition~\ref{prop:relaxed} states that, when restricting the search for privacy-preserving policies over the set of policies $\mathcal{F}:=\{\revise{\tilde{f}}|\exists g\in\mathcal{G}: \revise{\tilde{f}}=g\circ f\}$ for some set $\mathcal{G}$, the privacy metric can be relaxed to $I_\star(f(X);g(f(X)))$. Thus, }
the optimization problem in~\eqref{eqn:Problem1} with privacy measure~\eqref{eqn:define:P2} can be \textit{relaxed} to:
\begin{subequations} \label{eqn:Problem1_relaxed}
\begin{align}
\mathbf{P}'_\gamma: \min_{g\in\mathcal{G}} \quad & I_\star(Y;g(Y)),\\
\mathrm{s.t.}\hspace{.14in} & \essup_{y\in\range{Y}}\|y-g(y)\|_2\leq 1/\gamma,
\end{align}
\end{subequations}
\revise{In the relaxed problem, $f$ does not directly play a role in the privacy metric and, therefore, the optimal privacy-preserving policy becomes independent of $f$.} Note that such a relaxation is not possible for the privacy measure~\eqref{eqn:define:P1} because this measure of information is not symmetric and thus the data processing inequality does not hold for it in both directions. 



\section{Privacy-Preserving Policies} \label{sec:privacypolicy}
Before stating the results of the paper, the set of piecewise constant functions should be defined. Over the real line $\mathbb{R}$, a mapping $g:[\underline{y},\overline{y}]\rightarrow[\underline{y},\overline{y}]$ is a piecewise constant function if there exist $\underline{y}=a_1\leq a_2\leq \cdots\leq a_{q+1}=\overline{y}$ and $b_1\leq b_2\leq \cdots\leq b_{q}$ for some arbitrary number $q\in\mathbb{N}$ such that $g(y)=b_i$ for all $y\in [a_i,a_{i+1})$ except for $i=q$ in which case $g(y)=b_q$ for all $y\in [a_q,a_{q+1}]$. The ordered sets $(a_i)_{i=1}^{q+1}$ and $(b_i)_{i=1}^{q}$  are referred to as the parameters of the piecewise constant function. Let $\mathcal{Q}([\underline{y},\overline{y}])$ denote the set of all piecewise constant functions. For more general domains $\mathcal{X}$, a function $g:\mathcal{X}\rightarrow\mathbb{R}$ is a piecewise constant function if there exist sets $\{\mathcal{X}_i\}_{i=1}^q$ such that $\mathcal{X}\subseteq\bigcup_{i=1}^q \mathcal{X}_i$, $\mathcal{X}_i\cap \mathcal{X}_j=\emptyset$ if $i\neq j$, and $g(x)=b_i$ if $b_i\in\mathcal{X}_i$. The  ordered sets $(\mathcal{X}_i)_{i=1}^{q}$ and $(b_i)_{i=1}^{q}$  are referred to as the parameters of the piecewise constant function. Let $\mathcal{Q}(\mathcal{X})$ denote the set of all piecewise constant functions. When $\mathcal{X}$ is obvious from the context, $\mathcal{Q}$ is used instead of $\mathcal{Q}(\mathcal{X})$. The set of  piecewise constant functions is dense in $L^p$ for all $p\in[1,+\infty)$~\cite{levy2012elements}. In the next theorem, it is shown that searching over the set of piecewise constant functions is enough for finding the solution of~\eqref{eqn:Problem1}.

\begin{theorem}[Solution Class] The solution of~\eqref{eqn:Problem1} for the privacy metric in~\eqref{eqn:define:P2} over the set of piecewise differentiable functions is a piecewise constant function.
\end{theorem}

\begin{proof} Let $x\in\range{X}$ be any point such that $\nabla \revise{\tilde{f}}(x)\neq 0$. Then there exists a direction $d$ such that $d^\top\nabla \revise{\tilde{f}}(x)\neq 0$. Assume that $d^\top\nabla \revise{\tilde{f}}(x)> 0$; the proof for the other case is identical and is thus omitted.  By piecewise continuity of the derivatives, it can be deduced that there exists a small enough neighbourhood around $x$ of the form $\|\tilde{x}-x\|\leq \epsilon\|d\|$ inside which $d^\top\nabla \revise{\tilde{f}}(\tilde{x})>0$. Therefore, for all $w\in(-\epsilon,\epsilon)$, $\revise{\tilde{f}}(x+wd)$ is increasing and takes a unique value for any $w\in(-\epsilon,\epsilon)$. It must be established that no two distinct points in $\{(x+wd,\revise{\tilde{f}}(x+wd))|w\in(-\epsilon,\epsilon)\}$ are taxicab connected. This is done by contrapositive. Assume that this not the case. Therefore, there exists $(x,y),(x',y')\in\{(x+wd,\revise{\tilde{f}}(x+wd))|w\in(-\epsilon,\epsilon)\}\subseteq\range{X,Y}$ that are taxicab connected. This implies that there exists a sequence of points $\{(x_i,y_i)\}_{i=1}^n\subseteq\range{X,Y}$ such that $(x,y)\neq (x',y')$, $(x_1,y_1)=(x,y)$, $(x_n,y_n)=(x',y')$, and either $x_i=x_{i-1}$ or $y_i=y_{i-1}$ for all $i\in\{2,\dots,n\}$. Because $f$ is a function (i.e., $y_i=\revise{\tilde{f}}(x_i)=\revise{\tilde{f}}(x_{i-1})=y_{i-1}$ if $x_i=x_{i-1}$), all transitions such that $x_i=x_{i-1}$ can be eliminated (as it would also implies that $y_i=y_{i-1}$). Therefore,  a subsequence of points $\{(\bar{x}_i,\bar{y}_i)\}_{i=1}^{\bar{n}}\subseteq\{(x_i,y_i)\}_{i=1}^n\subseteq\range{X,Y}$ can be constructed so that $(\bar{x}_1,\bar{y}_1)=(x,y)$, $(\bar{x}_{\bar{n}},\bar{y}_{\bar{n}})=(x',y')$, and $\bar{y}_i=\bar{y}_{i-1}$ for all $i\in\{2,\dots,\bar{n}\}$. This implies that $y'=\bar{y}_{\bar{n}}=\bar{y}_{\bar{n}-1}=\cdots=\bar{y}_2=\bar{y}_1=y.$ This is in contradiction with the assumption that $(x,y)\neq (x',y')$ because it must be that $y\neq y'$; note that if $x_1\neq x_2$ in $\{(x+wd,\revise{\tilde{f}}(x+wd))|w\in(-\epsilon,\epsilon)\}$, it must also hold that $y_1\neq y_2$. Noting that no two distinct points in $\{(x+wd,\revise{\tilde{f}}(x+wd))|w\in(-\epsilon,\epsilon)\}$ are taxicab connected, there needs to be, at least, as many taxicab partitions as in the number of points in $\{(x+wd,\revise{\tilde{f}}(x+wd))|w\in(-\epsilon,\epsilon)\}$. This implies that $|\mathfrak{F}(X,Y)|=\infty$. The other category of functions is all functions for which $\nabla \revise{\tilde{f}}(x)= 0$ (where defined) for all $x$. The only functions that satisfy this condition are  piecewise constant functions. For piecewise constants $|\mathfrak{F}(X,Y)|=q<\infty$ with $q$ denoting the number of disjoint sets $\{\mathcal{X}_i\}_{i=1}^q$. 
\end{proof}

This fundamental result restricts the set of optimal privacy-preserving policies greatly and thus reduces the complexity of finding one. 

\begin{definition}[Uniform Quantizer] A uniform quantizer is a scalar piecewise constant function with parameters $(a_i)_{i=1}^{q+1}$ and $(b_i)_{i=1}^{q}$ such that  $a_{i+1}-a_{i}=a_{j+1}-a_j$ and $b_j=(a_j+a_{j+1})/2$ for all $1\leq i,j\leq q$. A uniform quantizer can be equivalently represented by the range $[a_1,a_{q+1}]$ and the number of bins $q$. 
\end{definition}

As the first step, the relaxed problem in~\eqref{eqn:Problem1_relaxed} is solved for scalar cases in the next theorem. 

\begin{theorem}[Relaxed Policy] \label{tho:Problem1_relaxed} Assume that $\range{Y}=[\underline{y},\overline{y}]\subseteq\mathbb{R}$. The solution of~\eqref{eqn:Problem1_relaxed} over $\mathcal{F}=\mathcal{Q}\circ\{f\}$ is a uniform quantizer, equi-dividing $\range{Y}$ into $\lceil \gamma(\overline{y}-\underline{y})/2\rceil$ bins.
\end{theorem}

\begin{proof}
Note that, for any $\revise{\tilde{f}}\in\mathcal{F}$,
\begin{align*}
\frac{1}{\mathfrak{Q}(\revise{\tilde{f}})}
&=\essup_{x\in\range{X}} |f(x)-g(f(x))|\\
&=
\essup_{y\in\range{Y}} |y-g(y)|\\
&=\max_{1\leq i\leq q}\max(|b_i-a_i|,|b_i-a_{i+1}|),
\end{align*}
where $g$ is any function in $\mathcal{Q}$. Furthermore, $I_\star(Y;g(Y))=q$. 
The problem~\eqref{eqn:Problem1_relaxed} can be rewritten as
\begin{align*}
\min_{(a_i)_{i=1}^{q+1},(b_i)_{i=1}^{q}} & q,\\
\mathrm{s.t.}\hspace{.27in} & \max_{1\leq i\leq q}\max(|b_i-a_i|,|b_i-a_{i+1}|)\leq \frac{1}{\gamma},\\
& a_{q+1}=\overline{y},\quad a_{1}=\underline{y}.
\end{align*}
By selecting $b_i=(a_i+a_{i+1})/2$, $\max(|b_i-a_i|,|b_i-a_{i+1}|)$ can be made as small as possible. Thus, this problem can be rewritten as 
\begin{subequations} \label{eqn:intermediateproblem}
\begin{align}
\min_{(a_i)_{i=1}^{q+1},(b_i)_{i=1}^{q}} & q,\\
\mathrm{s.t.}\hspace{.27in} & \max_{1\leq i\leq q} \frac{1}{2}|a_{i+1}-a_{i}|\leq \frac{1}{\gamma},\\
&\sum_{i=1}^q |a_{i+1}-a_{i}|=\overline{y}-\underline{y}.
\end{align}
\end{subequations}
It is easy to show that $q< \gamma(\overline{y}-\underline{y})/2$, the problem is not feasible. This is because
\begin{align*}
\sum_{i=1}^q |a_{i+1}-a_{i}|
&\leq q\max_{1\leq i\leq q}|a_{i+1}-a_{i}|\\
&\leq q2/\gamma\\
&<\overline{y}-\underline{y}.
\end{align*}
Therefore, a lower bound on the solution of~\eqref{eqn:intermediateproblem} is then the smallest integer that is larger than $\gamma(\overline{y}-\underline{y})/2$, i.e., $\lceil\gamma(\overline{y}-\underline{y})/2\rceil$. The uniform quantizer in the statement of theorem achieves the lower bound. 
%
\end{proof}

Now, the general problem in~\eqref{eqn:Problem1} can be considered for scalar queries over the set of piecewise continuous functions.

\begin{theorem}[Optimal Policy] \label{tho:Problem1_general} Assume that $\range{Y}\subseteq\mathbb{R}$. The solution of~\eqref{eqn:Problem1} for privacy measures in~\eqref{eqn:define:P1} over $\mathcal{F}=\mathcal{Q}(\range{X})$ is given by
\begin{subequations}
\begin{align}
b_i^*\in\argmin_{b_i}\max_{x\in\mathcal{X}^*_i} \;|f(&x)-b_i|, \\
\hspace{-.1in}\{\mathcal{X}_i^*\}_{i=1}^{q^*}\in \hspace{-.1in}\argmax_{\{\mathcal{X}_i\}_{i=1}^q:\range{X}\subseteq\bigcup_{i=1}^q \mathcal{X}_i} &  \; \min_{1\leq i\leq q}  \mu(\mathcal{X}_i),\\
\mathrm{s.t.}\hspace{.44in} &\;  \max_{1\leq i\leq q}\mathrm{rad}(f(\mathcal{X}_i))\leq \frac{1}{\gamma}.
\end{align}
\end{subequations}
For privacy measures in~\eqref{eqn:define:P2} over $\mathcal{F}=\mathcal{Q}(\range{X})$ is given by
\begin{subequations}
\begin{align}
b_i^*\in\argmin_{b_i}\max_{x\in\mathcal{X}^*_i} \;|f(&x)-b_i|, \\
\hspace{-.1in}\{\mathcal{X}_i^*\}_{i=1}^{q^*}\in \hspace{-.1in}\argmax_{\{\mathcal{X}_i\}_{i=1}^q:\range{X}\subseteq\bigcup_{i=1}^q \mathcal{X}_i} &  \; q,\\
\mathrm{s.t.}\hspace{.44in} &\;  \max_{1\leq i\leq q}\mathrm{rad}(f(\mathcal{X}_i))\leq \frac{1}{\gamma}.
\end{align}
\end{subequations}
\end{theorem}

\begin{proof} 
Note that, for any $\revise{\tilde{f}}\in\mathcal{F}=\mathcal{Q}(\range{X})$, there exists $\{\mathcal{X}_i,b_i\}_{i=1}^q$ such that $\range{X}\subseteq\bigcup_{i=1}^q \mathcal{X}_i$, $\mathcal{X}_i\cap \mathcal{X}_j=\emptyset$ if $i\neq j$, and $\revise{\tilde{f}}(x)=b_i$ if $b_i\in\mathcal{X}_i$. Hence,
\begin{align*}
\frac{1}{\mathfrak{Q}(f)}
&=\essup_{x\in\range{X}} |f(x)-\revise{\tilde{f}}(x)|\\
&=\max_{1\leq i\leq q}\sup_{x\in\mathcal{X}_i} |f(x)-b_i|.
\end{align*}
Let us consider the privacy measure in~\eqref{eqn:define:P1}. It can be shown that
\begin{align*}
d_0(X|\revise{\tilde{f}}(X))
&=\essinf_{x\in\range{X}} \log(\mu(\range{X|\revise{\tilde{f}}(x)}))\\
&=\essinf_{1\leq i\leq q} \log(\mu(\range{X|\revise{\tilde{f}}(x)=b_i}))\\
&=\min_{1\leq i\leq q}  \log(\mu(\mathcal{X}_i)).
\end{align*}
The problem~\eqref{eqn:Problem1} can be rewritten as
\begin{align*}
\max_{\{\mathcal{X}_i,b_i\}_{i=1}^q} & \min_{1\leq i\leq q}  \log(\mu(\mathcal{X}_i)),\\
\mathrm{s.t.}\hspace{.15in} & \max_{1\leq i\leq q}\sup_{x\in\mathcal{X}_i} |f(x)-b_i|\leq \frac{1}{\gamma},\\
& \range{X}\subseteq\bigcup_{i=1}^q \mathcal{X}_i.
\end{align*}
This problem can be rewritten again as
\begin{align*}
\max_{\{\mathcal{X}_i\}_{i=1}^q:\range{X}\subseteq\bigcup_{i=1}^q \mathcal{X}_i} &  \quad \min_{1\leq i\leq q}  \mu(\mathcal{X}_i),\\
\mathrm{s.t.}\hspace{.44in} &\quad  \max_{1\leq i\leq q}\min_{b_i}\sup_{x\in\mathcal{X}_i} |f(x)-b_i|\leq \frac{1}{\gamma}.
\end{align*}
Noting that $\mathrm{rad}(f(\mathcal{X}_i))=\min_{b_i}\sup_{y\in f(\mathcal{X}_i)} |y-b_i|$ concludes the proof for the first part. 
Now, let us consider the privacy measure in~\eqref{eqn:define:P2}. It can be seen that $I_\star(X;\revise{\tilde{f}}(X))=q$. This is because $(\mathcal{X}_i\times\{b_i\})_{i=1}^q$ forms a taxicab partition for $\range{X,f(X)}$. 
Hence, the problem~\eqref{eqn:Problem1} can be rewritten as
\begin{align} \label{eqn:intermediateproblem:2}
\min_{\{\mathcal{X}_i\}_{i=1}^q:\range{X}\subseteq\bigcup_{i=1}^q \mathcal{X}_i} &  \quad q,\\
\mathrm{s.t.}\hspace{.44in} &\quad  \max_{1\leq i\leq q}\min_{b_i}\sup_{x\in\mathcal{X}_i} |f(x)-b_i|\leq \frac{1}{\gamma}.
\end{align}
This concludes the proof.
\end{proof}

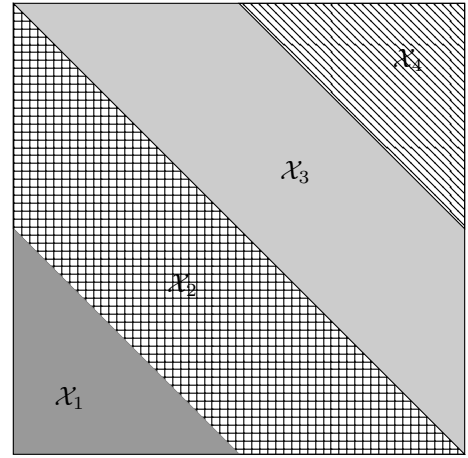
\begin{figure}
\centering
\begin{tikzpicture}[scale=1.5]
\draw[-,fill=black!20,color=black!20] (-2,2) -- (2,-2) -- (2,0) -- (0,2) -- cycle;
\draw[pattern=north west lines,pattern color=black] (2,0) -- (0,2) -- (2,2) -- cycle;
\draw[pattern=grid,pattern color=black] (-2,0) -- (-2,2) -- (2,-2) -- (0,-2) -- cycle;
\draw[-,fill=black!40,color=black!40] (-2,0) -- (0,-2) -- (-2,-2) -- cycle;
\draw[-] (-2,-2) -- (2,-2) -- (2,2) -- (-2,2) -- cycle;
\node[] at (-1.5,-1.5) {$\mathcal{X}_1$};
\node[] at (-0.5,-0.5) {$\mathcal{X}_2$};
\node[] at (0.5,0.5) {$\mathcal{X}_3$};
\node[] at (1.5,1.5) {$\mathcal{X}_4$};
\end{tikzpicture}
\caption{\label{fig:1} The regions $\{\mathcal{X}_i\}_{i=1}^4$ for the optimal privacy-preserving policy in Theorem~\ref{tho:Problem1_general} for $\range{X}=[-2,2]^2$, $\gamma=2$, and linear query $f(x)=\mathds{1}^\top x/2$. For the optimal policy, $b_1=-1.5$, $b_2=-0.5$, $b_3=0.5$, and $b_4=1.5$. }
\end{figure}

For the case where $\range{X}\subseteq\mathbb{R}$, the results of Theorems~\ref{tho:Problem1_general} and~\ref{tho:Problem1_relaxed} \revise{are equal}~\cite{KONNO1988205}. Therefore, there is no loss of generality in designing the quantizer after computing $f(x)$ rather than designing a general $\revise{\tilde{f}}(x)$. In the next corollary, this property  is proved for general queries under mild assumptions. 

\begin{corollary} \label{cor:equivalence} Let $f$ be a function that $f^{-1}(y):=\{x|f(x)=y\}$ is a connected set for all $y\in\range{Y}$. Then, the optimal policy in Theorem~\ref{tho:Problem1_general} for the privacy metric~\eqref{eqn:define:P2} is equal to the the optimal policy in Theorem~\ref{tho:Problem1_relaxed}.
\end{corollary}

\begin{proof} The solution of~\eqref{eqn:Problem1} for the privacy measure in~\eqref{eqn:define:P2} is given by~\eqref{eqn:intermediateproblem:2}. Define $\mathcal{Y}'_i=\{y|\exists x\in\mathcal{X}_i:y=f(x)\}$. The inequality constraint in~\eqref{eqn:intermediateproblem:2} is equivalent to saying that that $ \max_{1\leq i\leq q}\min_{b_i}\sup_{y\in\mathcal{Y}'_i} |y-b_i|\leq 1/\gamma$. Let $\mathcal{Y}_i$ be defined such that $\mathcal{Y}_1=\mathcal{Y}'_1$ and $\mathcal{Y}_i=\mathcal{Y}'_i\setminus (\mathcal{Y}_1\cup\cdots\cup \mathcal{Y}_{i-1})$  for all $i>1$. Clearly, $\mathcal{Y}_i\subseteq\mathcal{Y}'_i$ and thus $ \max_{1\leq i\leq q}\min_{b_i}\sup_{y\in\mathcal{Y}_i} |y-b_i|\leq 1/\gamma$. If $\mathcal{Y}_i$ is connected, it should take one of the following forms $[a_i,a_{i+1}]$, $[a_i,a_{i+1})$, $(a_i,a_{i+1}]$, or $(a_i,a_{i+1})$. Therefore, by selecting $b_i=(a_i+a_{i+1})/2$ minimizes $\sup_{y\in\mathcal{Y}_i} |y-b_i|$. This implies that~\eqref{eqn:intermediateproblem:2} can be rewritten as the optimization problem in the statement of Theorem~\ref{tho:Problem1_relaxed}. 
\end{proof}

\begin{figure}
\centering
\begin{tikzpicture}[scale=1.5]
\clip (-2.01,-2.01) rectangle (2.01,2.01);
\draw[preaction={fill,white},pattern=vertical lines,pattern color=black] 
(0,0) ellipse (3.4641 and 2.4495);
\draw[-,fill=black!60] 
(0,0) ellipse (3.3166 and 2.3452);
\draw[preaction={fill,white},pattern=north east lines,pattern color=black] 
(0,0) ellipse (3.1623 and 2.2361);
\draw[-,fill=black!50]  
(0,0) ellipse (3.0000 and 2.1213);
\draw[preaction={fill,white},pattern=crosshatch,pattern color=black] 
(0,0) ellipse (2.8284 and 2.000);
\draw[-,fill=black!40] 
(0,0) ellipse (2.6458 and 1.8708);
\draw[preaction={fill,white},pattern=crosshatch dots,pattern color=black] 
(0,0) ellipse (2.4495 and 1.7321);
\draw[-,fill=black!30] 
(0,0) ellipse (2.2361 and 1.5811);
\draw[preaction={fill,white},pattern=grid,pattern color=black] 
(0,0) ellipse (2.0000 and 1.4142);
\draw[-,fill=black!20] 
(0,0) ellipse (1.7321 and 1.2247);
\draw[preaction={fill,white},pattern=north west lines,pattern color=black] 
(0,0) ellipse (1.4142 and 1.0000);
\draw[-,fill=black!10] 
(0,0) ellipse (1.0000 and 0.7071);
\draw[-] (-2,-2) -- (2,-2) -- (2,2) -- (-2,2) -- cycle;
\end{tikzpicture}
\begin{tikzpicture}[overlay,>=stealth,shift={(-3.1,+3)},scale=1.5]
\node[] at (0,0) {$\mathcal{X}_1$};
\node[] at (2.8,0) (x3) {$\mathcal{X}_3$};
\draw[->] (x3) -- (1.5,0);
\node[] at (2.7,0.7) (x4) {$\mathcal{X}_4$};
\draw[->] (x4) -- (1.7,.3);
\node[] at (2.5,1.5) (x5) {$\mathcal{X}_5$};
\draw[->] (x5) -- (1.8,.7);
\node[] at (-2.8,0) (x6) {$\mathcal{X}_6$};
\draw[->] (x6) -- (-1.85,1.05);
\node[] at (-2.7,0.7) (x7) {$\mathcal{X}_7$};
\draw[->] (x7) -- (-1.85,1.25);
\node[] at (-2.5,1.5) (x8) {$\mathcal{X}_8$};
\draw[->] (x8) -- (-1.9,1.4);
\node[] at (-2.7,-0.7) (x9) {$\mathcal{X}_9$};
\draw[->] (x9) -- (-1.9,-1.6);
\node[] at (-2.7,-1.5) (x10) {$\mathcal{X}_{10}$};
\draw[->] (x10) -- (-1.85,-1.75);
\node[] at (-2.7,-1.9) (x11) {$\mathcal{X}_{11}$};
\draw[->] (x11) -- (-1.9,-1.85);
\node[] at (2.8,-0.9) (x2) {$\mathcal{X}_2$};
\draw[->] (x2) -- (1.04,-.3);
\node[] at (2.7,-1.9) (x12) {$\mathcal{X}_{12}$};
\draw[->] (x12) -- (1.93,-1.97);
\end{tikzpicture}
\caption{\label{fig:2} The regions $\{\mathcal{X}_i\}_{i=1}^{12}$ for the  optimal privacy-preserving policy in Theorem~\ref{tho:Problem1_general} for $\range{X}=[-2,2]^2$, $\gamma=2$, and nonlinear query $f(x)=x^\top \mathrm{diag}(1,2)x$. For the optimal policy, $b_i=i-0.5$ for all $1\leq i\leq 12$. }
\end{figure}
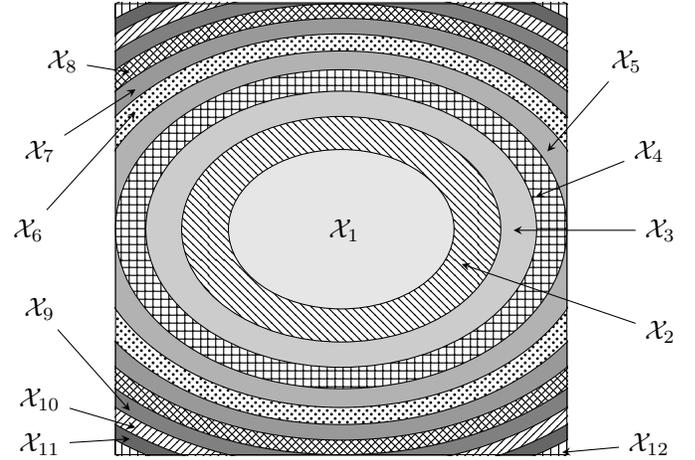

\begin{example}
Consider a simple example in which reporting the average of two real numbers in $[-2,2]$ is of interest. Therefore, the query is $f(x)=\mathds{1}^\top x/2$. First, consider the relaxed problem in~\eqref{eqn:Problem1_relaxed}. Assume that $\gamma=2$. The optimal policy in this case is to quantize $y$ with a uniform quantizer over $[-2,2]$ with $4$ bins, denoted by $g(\cdot)$. Thus,
\begin{align*}
\revise{\tilde{f}}(x)=g(f(x))=
\begin{cases}
-1.5, & -2\leq f(x)<-1, \\
-0.5, & -1\leq f(x)<0,\\
0.5, & 0\leq f(x)<1,\\
1.5, & 1\leq f(x)\leq 2.
\end{cases}
\end{align*}
This function can be rewritten as
\begin{align} \label{eqn:optimal_relaxed_example}
\revise{\tilde{f}}(x)=
\begin{cases}
-1.5, & -4\leq x_1+x_2<-2, \\
-0.5, & -2\leq x_1+x_2<0,\\
0.5, & 0\leq x_1+x_2<2,\\
1.5, & 2\leq x_1+x_2\leq 4,
\end{cases}
\end{align}
where $x_i$ denotes the $i$-th entry of $x$. Now, Theorem~\ref{tho:Problem1_general} can be used to find the optimal privacy-preserving policy for the case with privacy metric in~\eqref{eqn:define:P2}. Figure~\ref{fig:1} illustrates the regions $\{\mathcal{X}_i\}_{i=1}^4$ for the optimal privacy-preserving policy in Theorem~\ref{tho:Problem1_general} for $\range{X}=[-2,2]^2$, $\gamma=2$, and linear query $f(x)=\mathds{1}^\top x/2$. For the optimal policy in Figure~\ref{fig:1}, $b_1=-1.5$, $b_2=-0.5$, $b_3=0.5$, and $b_4=1.5$. It is interesting to note that the optimal policy in Figure~\ref{fig:1} is in fact equal to~\eqref{eqn:optimal_relaxed_example}. Therefore, the relaxation in~\eqref{eqn:Problem1_relaxed} is without loss of generality in this example. This is because $f$ meets the condition of Corollary~\ref{cor:equivalence}. 

\begin{figure}
\centering
\begin{tikzpicture}
\node[] at (0,0) {
\psfrag{x1}[cc][][0.6][0]{$10^{-1}$}
\psfrag{x2}[cc][][0.6][0]{$10^{\,0}$}
\psfrag{x3}[cc][][0.6][0]{$10^{\,1}$}
\psfrag{y1}[cc][][0.6][0]{\quad $0$}
\psfrag{y2}[cc][][0.6][0]{\quad $1$}
\psfrag{y3}[cc][][0.6][0]{\quad $2$}
\psfrag{y4}[cc][][0.6][0]{\quad $3$}
\psfrag{y5}[cc][][0.6][0]{\quad $4$}
\psfrag{z1}[cc][][0.6][0]{\quad $0$}
\psfrag{z2}[cc][][0.6][0]{\quad $2$}
\psfrag{z3}[cc][][0.6][0]{\quad $4$}
\psfrag{z4}[cc][][0.6][0]{\quad $6$}
\psfrag{z5}[cc][][0.6][0]{\quad $8$}
\psfrag{w1}[cc][][0.6][0]{\quad $0$}
\psfrag{w2}[cc][][0.6][0]{\quad $1$}
\psfrag{w3}[cc][][0.6][0]{\quad $2$}
\psfrag{w4}[cc][][0.6][0]{\quad $3$}
\psfrag{w5}[cc][][0.6][0]{\quad $4$}
\includegraphics[width=\linewidth]{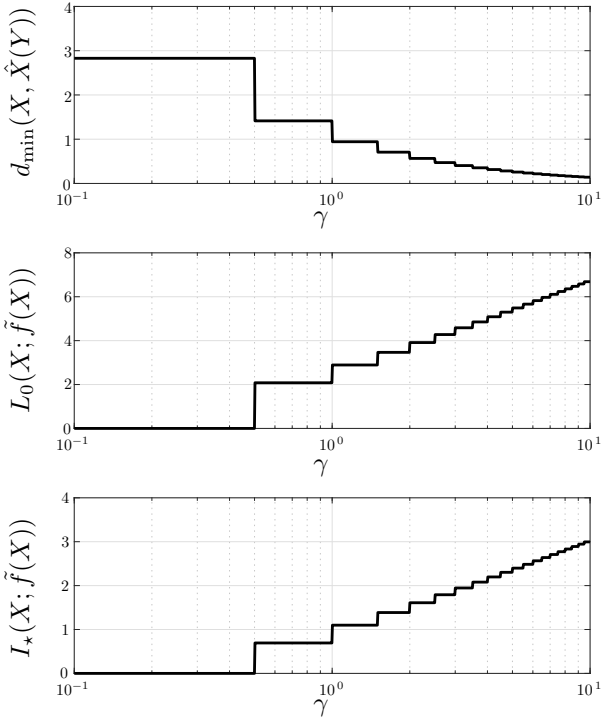}};
\node[] at (0,-4.7) {$\gamma$};
\node[] at (0,-1.5) {$\gamma$};
\node[] at (0,+1.8) {$\gamma$};
\node[rotate=90] at (-4,+0.2) {$L_0(X;\tilde{f}(X))$};
\node[rotate=90] at (-4,+3.5) {$d_{\min}(X,\hat{X}(Y))$};
\node[rotate=90] at (-4,-3.1) {$I_\star(X;\tilde{f}(X))$};
\end{tikzpicture}
\vspace{-.3in}
\label{fig:gamma}
\caption{The effect of constant $\gamma$ on the trade-off between utility $\gamma$ and privacy, captured using the amount of the private information leaked $L_0(X;\tilde{f}(X))$ and $I_\star(X;\tilde{f}(X))$ as well as the adversary's  estimation error $d_{\min}(X,\hat{X}(Y))$. }
\end{figure}

\revise{In~\eqref{eqn:Problem1},  $\gamma$ changes the trade-off between privacy and utility:  by increasing $\gamma$, a larger bound on the quality is required  and privacy guarantee must be weakened. To demonstrate this, consider the example discussed above with a general $\gamma>0$. The optimal policy in this case is to quantize $y$ with a uniform quantizer over $[-2,2]$ with $\lceil 2\gamma \rceil$ bins. Figure~\ref{fig:gamma} illustrates the amount of the private information leaked $L_0(X;\tilde{f}(X))$ and  $I_\star(X;\tilde{f}(X))$, and  the adversary's estimation error $d_{\min}(X,\hat{X}(Y))$. Evidently, as the quality improves ($\gamma$ increases), the privacy guarantee weakens (the amount of the leaked information increases and the adversary's estimation error decreases). }

Now, focus on a non-linear query of the form $f(x)=x_1^2+2x_2^2$. In this case, $\range{Y}=[0,12]$. Therefore, the optimal policy of the relaxed problem in~\eqref{eqn:Problem1_relaxed} for $\gamma=2$ is a uniform quantizer over $[0,12]$ with $12$ bins. Again, use $g$ denote this quantizer. It can be seen that
\begin{align} \label{eqn:optimal_relaxed_example2}
\revise{\tilde{f}}(x)=
i+0.5, \;  i\hspace{-.03in}\leq \hspace{-.03in}x_1^2\hspace{-.03in}+\hspace{-.03in}2x_2^2\hspace{-.03in}<\hspace{-.03in}i+1, \; \forall i\in\{0,\dots,11\},
\end{align}
Again, Theorem~\ref{tho:Problem1_general} can be used to find the optimal privacy-preserving policy in this case. Figure~\ref{fig:2} illustrates the regions $\{\mathcal{X}_i\}_{i=1}^{12}$ for the optimal privacy-preserving policy in Theorem~\ref{tho:Problem1_general} for $\range{X}=[-2,2]^2$, $\gamma=2$, and non-linear query $f(x)=x^\top \mathrm{diag}(1,2)x$. For the optimal policy, $b_i=i-0.5$ for all $1\leq i\leq 12$. Similarly,  the optimal policy in \revise{Figure~\ref{fig:2}} is equal to~\eqref{eqn:optimal_relaxed_example2} and thus, the relaxation in~\eqref{eqn:Problem1_relaxed} is again without loss of generality as $f$ meets the condition of Corollary~\ref{cor:equivalence}.
\end{example}

\revise{
\begin{example} Consider a practical example in which the private dataset contains the height of $n_x$ individuals in the range of $[100,250]$ centimetres. The submitted query is to compute the average height of the individuals in the dataset, i.e., $f(x)=\mathds{1}^\top x/n_x$. Following the results of the paper,  the optimal privacy-preserving policy is to quantize $f(x)$ using a uniform quantizer over $[100,250]$ with $\lceil 75\gamma\rceil$ bins. In this case, $d_{\min}(X,\hat{X}(\tilde{f}(X)))= 150\sqrt{2}/\lceil 75\gamma\rceil$, which is independent of $n_x$. This is because the worst-case in terms of preserving privacy occurs in a society with $n_x-2$ individuals whose heights are equal to $250$ and two individuals whose heights are within $(250-150/\lceil 75\gamma\rceil,250]$. To be able to guarantee an error of at least $10$ centimetres for the adversary, $\gamma$ must be selected to be larger than $22/75\approx 0.2933$. 
\end{example}
}

\section{Relationship to Other Notions of Privacy} \label{sec:kanonimity}
In this section, the privacy credentials of $k$-anonymity is analyzed using the measures of privacy in~\eqref{eqn:define:P12}. Consider a dataset $x\in \mathbb{X}\subseteq\mathbb{R}^{n\times m}$ with $n$ rows (entries or individuals) and $m$ columns (attributes). The following argument can easily be extended to other sets and is thus without loss of generality. 

\begin{definition}[$k$-anonymity~\cite{samarati1998protecting,sweeney2002k, samarati2001protecting}] A release of data is said to have the $k$-anonymity property if the information for each individual contained in the release cannot be distinguished from at least $k-1$ individuals whose information also appear in the release.
\end{definition}

\begin{proposition} \label{prop:kanonymity} There exists a reporting function $\revise{\tilde{f}}(X)$ admitting $k$-anonymity property for which the following holds: 
\begin{itemize}
\item $d_0(X|f(X))=0$ (and thus $L_0(X;f(X))=h_0(X)$);
\item $I_\star(X;f(X))=\infty$. 
\end{itemize}
\end{proposition}

\begin{proof}
Consider the case where $x$ is a dataset that has $k$ identical individuals.
Let the first $k$ rows denote the identical individuals. This is without the loss of generality as otherwise the rows can be swapped. Let $f:\mathbb{R}^{n\times m}\rightarrow\mathbb{R}^{n\times m}$ be any $k$-anonymous reporting function. Assume that the $i$-th row of $f(x)$ is report corresponding to the $i$-th row of $x$. This is again without the loss of generality as otherwise the output rows can be swapped. Construct $\revise{\tilde{f}}$ such that 
\begin{align*}
\revise{\tilde{f}}(x)=
\begin{bmatrix}
x_1 \\
\vdots \\
x_k \\
\begin{bmatrix}
0_{(n-k)\times n} & I_{n-k}
\end{bmatrix}
f(x)
\end{bmatrix}.
\end{align*}
By construction $\revise{\tilde{f}}$ is also a $k$-anonymous reporting function. However, 
\begin{align*}
\range{X|\revise{\tilde{f}}(x)}=\range{X|f(x)}\cap \left\{\begin{bmatrix}
w \\
z
\end{bmatrix}\in\mathbb{X} \,\bigg|\, w=\begin{bmatrix}
x_1 \\
\vdots \\
x_k
\end{bmatrix} \right\},
\end{align*}
which shows that $\mu(\range{X|\revise{\tilde{f}}(x)})=0$. Thus, $d_0(X|\revise{\tilde{f}}(X))=0$. Finally, noting that $\range{X|\revise{\tilde{f}}(x)}$ must be included in the taxicab partitions for all choices of $x_1=\dots=x_k$, $|\mathfrak{F}(X,\revise{\tilde{f}}(X))|=+\infty$. This shows that $I_\star(X;\revise{\tilde{f}}(X))=+\infty$. 
\end{proof}

Proposition~\ref{prop:kanonymity} shows that $k$-anonymity is not private. This is because of the homogeneity attack~\cite{1617392}, i.e., attacks that leverage the cases in which all the values for a sensitive value within a set of $k$ records are identical. In such cases, even though the data has been $k$-anonymized, the sensitive value for the set of $k$ records may be exactly predicted. Such cases are explored to prove Proposition~\ref{prop:kanonymity}.

\section{Conclusions and Future Work}
\label{sec:conclusions}

A deterministic privacy metric  using non-stochastic information theory was presented. It was assumed that anyone can submit a query to a trusted server with access to a non-stochastic uncertain private data. Optimal privacy-preserving policy was proved to be a quantized version of the output of the submitted query. Finally, it was proved that $k$-anonymity is not privacy-preserving using the proposed privacy metric. Future work can focus on analysing non-scalar queries as well as demonstrating the performance of the method on publicly available datasets.


\bibliographystyle{ieeetr}
\bibliography{sample-bibliography}

\end{document}